\documentclass[a4paper,11pt]{article}

\usepackage{fullpage}

\usepackage{cite} % sort references within the same \cite command
\usepackage{enumerate} % fancy options for enumerate environment
\usepackage{amsmath,amsfonts,amssymb,amsthm}
\usepackage{graphicx,epsfig,wrapfig}
\usepackage{url,hyperref}
\usepackage{color,xcolor}
\usepackage{microtype} % better line breaks and fewer overfull hboxes
\usepackage{fancybox}
\usepackage{stmaryrd}
\usepackage{xspace}

\graphicspath{{figs/}}

\newcommand{\RR}{\mathbb{R}}
\newcommand{\NN}{\mathbb{N}}

\newcommand{\R}{\mathcal{R}}
\renewcommand{\L}{\mathcal{L}}

\newcommand{\eps}{\varepsilon}

\def\area{\mathrm{area}}
\def\slab{\mathrm{slab}}
\def\bb{\mathrm{bb}}
\def\dist{\mathrm{dist}}
\def\top{\mathrm{top}}
\def\bot{\mathrm{bot}}
\def\optpoints{\mathrm{\kappa^*}}
\def\optarea{\mathrm{area^*}}

\newtheorem{theorem}{Theorem}
\newtheorem{lemma}[theorem]{Lemma}

\newtheorem{corollary}[theorem]{Corollary}

\usepackage{hyperref}
\hypersetup{%
  breaklinks,%
  ocgcolorlinks,
  colorlinks=true,%
  linkcolor=[rgb]{0.45,0.0,0.0},%
  citecolor=[rgb]{0,0,0.45}
}

\makeatletter
\partopsep\z@ \textfloatsep 10pt plus 1pt minus 4pt
\def\section{\@startsection {section}{1}{\z@}%
  {-3.5ex plus -1ex
    minus -.2ex}{2.3ex plus .2ex}{\large\bf}}
\def\subsection{\@startsection{subsection}{2}%
  {\z@}{-3.25ex plus
    -1ex minus -.2ex}{1.5ex plus .2ex}{\normalsize\bf}}
\def\@fnsymbol#1{\ensuremath{\ifcase#1\or 1\or 2\or 3\or 4\or
    5\or 6\or 7 \or 8\ or 9 \or 10\or 11 \else\@ctrerr\fi}}
\makeatother

\let\geq\geqslant
\let\leq\leqslant
\let\ge\geqslant
\let\le\leqslant

\title{Covering Many Points with a Small-Area Box}
\author{Mark de Berg\thanks{Department of Computer Science, TU
    Eindhoven, the Netherlands. 
    Supported by the Netherlands Organization 
    for Scientific Research (NWO) under project no.~024.002.003.}
  \and
  Sergio Cabello\thanks{Department of Mathematics, IMFM, and Department
    of Mathematics, FMF,  
    University of Ljubljana, Slovenia.
    Supported by the Slovenian Research Agency, program P1-0297 and projects J1-8130 and J1-8155.}
  \and 
  Otfried Cheong\thanks{School of Computing, KAIST, Korea.
    Supported  by ICT R\&D program of
    MSIP/IITP~[R0126-15-1108].}
  \and 
  David Eppstein\thanks{Computer Science Department, University of
    California, Irvine. 
    Supported in part by NSF grants CCF-1228639, 
    CCF-1618301, and CCF-1616248.}
  \and 
  Christian Knauer\thanks{Computer Science Department, University of
    Bayreuth, Germany. 
    Supported in part by DFG grant \mbox{Kn~591/3-3}.}
}

\begin{document}
\maketitle

\begin{abstract}
  Let $P$ be a set of $n$ points in the plane. We show how to find,
  for a given integer $k>0$, the smallest-area axis-parallel rectangle
  that covers $k$ points of $P$ in $O(nk^2 \log n+ n\log^2 n)$ time.
  We also consider the problem of, given a value $\alpha>0$, covering
  as many points of $P$ as possible with an axis-parallel rectangle of
  area at most $\alpha$. For this problem we give a probabilistic
  $(1-\eps)$-approximation that works in near-linear time: In
  $O((n/\eps^4)\log^3 n \log (1/\eps))$ time we find an axis-parallel
  rectangle of area at most $\alpha$ that, with high probability,
  covers at least $(1-\eps)\optpoints$ points, where $\optpoints$ is
  the maximum possible number of points that could be covered.
    
  \medskip
  \noindent\textbf{Keywords:} geometric optimization, covering,
  axis-parallel rectangle, near-linear time, approximation algorithm.
\end{abstract}

\section{Introduction}
In this paper we consider two closely related shape-fitting problems
for a given point set~$P$ in the plane. In both problems we are
searching for an axis-parallel rectangle, or a \emph{box} as we will
call it, and we are interested in the trade-off between the box area
and the number of points covered by the box. More precisely, we are
interested in the following two optimization problems.
\begin{itemize}
\item Given a set $P$ of points and an integer $k\ge 2$, find 
  \[ 
  \optarea(P,k) ~=~ \min \big\{ \area(R) \mathrel{\big|} \text{$R$ is a box
    with}~|R\cap P| \ge k \big\}.
  \]
  That is, we are interested in covering at least $k$ points of $P$
  with a box of minimum area.
\item Given a set $P$ of points and a real value $\alpha>0$, find 
  \[ 
  \optpoints(P,\alpha) ~=~ \max \big\{ |R\cap P|\mathrel{\big|}
  \text{$R$ is a box with}~\area(R)\le \alpha \big\}. 
  \]
  That is, we are interested in covering the maximum number of points
  of $P$ with a box of area at most $\alpha$.
\end{itemize}
The two problems are closely related because for all finite point sets $P$,
and all $k\in \NN$ and $\alpha\in \RR_{>0}$ we have 
\[
    \optarea(P,k)\le \alpha \Longleftrightarrow \optpoints(P,\alpha) \ge k.
\]
So the second problem can be solved using binary search
on~$k$ and a solution to the first problem.

When minimizing the area of the box covering $k$ points, the set of
optimal solutions is invariant under scaling of either of the axes.
This means that, if we consider any map $(x,y)\mapsto
\varphi(x,y)=(\alpha_1 x+\beta_1,\alpha_2 y+\beta_2)$ with
$\alpha_1,\alpha_2\neq 0$, then a box~$R$ is an optimal solution for
$\optarea(P,k)$ if and only if $\varphi(R)$ is a solution for
$\optarea(\varphi(P),k)$. Thus, minimizing the area is especially
useful when the units of each axis have incomparable meanings. In
contrast, in such a case it is meaningless to minimize the perimeter.

The problem of covering $k$ points with a minimum-area (or
minimum-perimeter) box was previously considered by Segal and
Kedem~\cite{segked-98}, who provided an algorithm suitable for values
of $k$ close to $n$, with running time $O(n + k(n-k)^2)$. In contrast,
we study the case when $k$ is small, so that it is preferable to
decrease the dependence on $n$ at the expense of increasing the
dependence on $k$. For the case of small $k$, several
papers~\cite{AhnBDDKKRS11,DasGN05,segked-98} erroneously claim that
previous algorithms of Aggarwal~et al.~\cite{aiks-91} and Eppstein and
Erickson~\cite{epperi-94} solve the problem in running time
$O(k^2n\log n)$ or $O(n\log n+k^2n)$, respectively. However, these
previous algorithms apply only to the minimum-perimeter version of the
problem. They do not work for the minimum-area version, because they
are based on the fact that for the minimum-perimeter version, the
optimal subset of $k$ points can be found among the $O(k)$ nearest
neighbors to one of the points---something which is not true for the
minimum-area version. The same obstacle appears when trying to extend
the algorithms of Datta et al.~\cite{DattaLSS95} from the
minimum-perimeter to the minimum-area problem. For the minimum-area
problem that we study here, we cannot restrict our attention to sets
of nearest neighbors, and must use alternative methods to obtain our
time bounds. The results in those papers do not depend on the
mistaken claim, only the attribution of previous work is incorrect.

After our preprint was made public~\cite{BergCCEK16}, Kaplan et
al.~\cite{KaplanRS17} showed that the problem of covering $k$ points
with a minimum-area box can be solved in $O(n^{5/2}\log^2 n)$ time.
This is the first subcubic algorithm and it is more efficient than
previous results for a large range of $k$. For the minimum-perimeter
problem, they provide an algorithm running in $O(nk^{3/2}\log k \log
n)$ time. However, as they note, one of the steps used in their
algorithm does not work for the minimum-area problem. The difficulty
is essentially as we mentioned above: For the minimum-area problem, we
do not know how to transform the problem into $O(n/k)$ instances of
$O(k)$ points each.

There have been several works on minimizing the size of the smallest
\emph{disk} that contains $k$~points. Here it does not matter whether
we minimize the area or the perimeter. \mbox{Har-Peled} and
Mazumdar~\cite{Har-PeledM05} give a randomized algorithm to find a
disk that contains $k$ points in $O(nk)$ expected time, improving the
works by Efrat, Eppstein, Erickson, Matou{\v s}ek, and
Sharir~\cite{EfratSZ94,epperi-94,Matousek95}. In follow-up work,
\mbox{Har-Peled} and Raichel~\cite{Har-PeledR15} aim for fast
$(1+\eps)$-approximations. Das et.~al~\cite{DasGN05} consider
covering $k$ points with rectangles of arbitrary orientations.
\medskip

The problems that we are interested in, where we want to find an
optimal box of arbitrary aspect ratio, are relatively easy if we make
certain assumptions about the input. For instance, if we were given
the aspect ratio of an optimal box, we could rescale one axis to
reduce the problem to finding an optimal square, a problem that is
very similar to the problem for disks. Similarly, if we had, say, a
$2$-approximation to the aspect ratio of the optimal box, then we
could rescale one axis and reduce the search to fat boxes. In this
scenario, finding a smallest square box gives a constant-factor
approximation to the optimum fat box, and using a grid approach,
like in \mbox{Har-Peled} and Mazumdar~\cite{Har-PeledM05}, we only need to
solve $O(n/k)$ instances of size $O(k)$, which can be done in roughly
$O(nk^2)$ time. Thus, we can search for the optimal box with constant
fatness in roughly $O(nk^2)$ time. Also, if we assume that the
coordinates are integers between $0$ and a bound $U$, this approach
allows us to compute $\optarea(P,k)$ in roughly $O(nk^2 \log U)$ time,
by trying $O(\log U)$ different aspect ratios in geometric
progression. The main goal of our paper is to avoid any such
assumptions, and to still get similar running times.

\paragraph*{Our results.}
Here is a summary of our main results and an overview of the approach.
Let $P$ be a given set of $n$ points in the plane.
\begin{enumerate}[(a)]
\item We show how to find, for a given integer $k>0$, the value
  $\optarea(P,k)$ in $O(nk^2 \log n+ n\log^2 n)$ time. Within the
  same time bound we can also construct an optimal solution, that is,
  a box that contains $k$ points of $P$ and has area $\optarea(P,k)$.
  This is the only known algorithm with a near-linear dependency on
  $n$; see the discussion above. To achieve this result, we use a
  divide-and-conquer approach to generate $O(n\log n)$ subproblems,
  each with $O(k)$ points, where we only have to consider boxes that
  contain a fixed point on the boundary. The divide-and-conquer
  method resembles that by Aronov et al.~\cite{AronovES10}. These
  results are presented in Section~\ref{sec:optarea}. In fact, here
  we solve a slightly more general problem to enable some improvements
  in the running time of the problem considered next in~(b).
  
\item We give a randomized algorithm that, for a given value
  $\alpha>0$ and a parameter $\eps\le 1/2$, with high probability runs
  in time $O((n/\eps^4)\log^3 n \log (1/\eps))$ and returns a box that
  has area $\alpha$ and covers at least $(1-\eps)\optpoints(P,\alpha)$
  points of~$P$. Note that the running time is~$O(n\log^3 n)$ when
  $\eps$ is fixed. An overview of the approach is as follows (a
  similar high-level approach is used for example
  in~\cite{AgarwalHRSSW02,AronovH08,CabelloDP13,BergCH09}). First, we
  find a $4$-approximation to the value $\optpoints(P,\alpha)$. Then
  we use a random sample $S$ of $P$ of appropriate size such that,
  with high probability, $\optpoints(S,\alpha)=\Theta((1/\eps^2)\log
  n)$ and an optimal solution for $\optpoints(S,\alpha)$ contains
  $(1-\eps)\optpoints(P,\alpha)$ points of $P$. To slightly improve
  the running time, we avoid computing $\optpoints(S,\alpha)$ exactly
  for the random sample $S$, as we can afford to use a
  $(1-\eps)$-approximation instead, by performing a binary search,
  where at each step we have to decide whether $\optarea(S,k')\le
  \alpha$ for some given $k'$. An additional slight improvement is
  obtained using the observation that~$S$ is always the same, but the
  test values~$k'$ in the binary search change. These results are
  described in Section~\ref{sec:points}.
\end{enumerate}

\paragraph*{Notation and conventions.}

As noted, a \emph{box} is an axis-parallel rectangle. For a box~$R$, let
$\top(R)$ and $\bot(R)$ be its top and bottom edge. For a point $p\in
\RR^2$, we use $p_x$ and $p_y$ for its $x$- and $y$-coordinate,
respectively.

We assume that the point set is in \emph{general position}, meaning
that no two points have the same $x$-coordinate or the same
$y$-coordinate. This can be enforced by a symbolic perturbation of the
points. For example, we can index the points as $p_1,\dots,p_n$ and
replace each point $p_i$ with the point $p_i+(i\cdot \eps,i\cdot\eps)$
for an infinitesimal value $\eps>0$. When minimizing the area of the
box covering $k$ points, we drop in the resulting area any terms that
depend on $\eps$. When maximizing the number of points to be covered,
we allow boxes of area $\alpha+n\eps\rho$, where $\rho$ is the
perimeter of the bounding box of~$P$.

\section{Minimizing area for a given number of points}
\label{sec:optarea}

We will use the following result for batched reporting in $2$-sided
rectangles. See Figure~\ref{fig:rangek} for an example.

\begin{figure}
  \centerline{\includegraphics{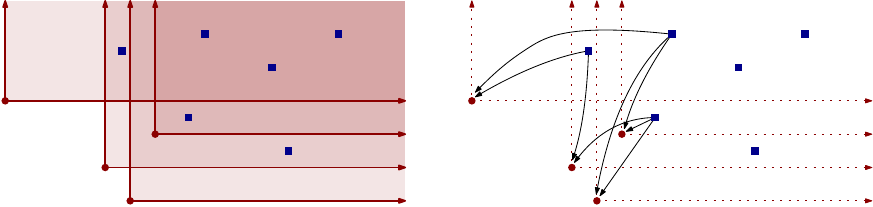}}
  \caption{Left: Example of the scenario considered in
    Lemma~\ref{lem:rangek}. Right: Points to be reported for each
    rectangle when $k$=2. An arrow indicates that the tail is reported
    for the head.}
  \label{fig:rangek}	
\end{figure}

\begin{lemma}
  \label{lem:rangek}
  Let $A$ and $B$ be sets of at most $n$ points in $\RR^2$. For each
  point $b\in B$, let $R_b$ be the 2-sided rectangle
  $[b_x,\infty)\times [b_y,\infty)$. In time $O(kn+ n \log n)$ we can
      find, for all $b\in B$, the $k$ points in $A\cap R_b$ with
      smallest $x$-coordinate.
\end{lemma}
\begin{proof}
  The result can be obtained in a standard manner, using a sweep-line
  algorithm that sweeps the plane with a vertical line $\ell$ from
  left to right. For completeness we give the details.
    
  Let $A_\ell$ and $B_\ell$ be the points to the left of~$\ell$ of $A$
  and $B$, respectively. Consider the family of rectangles
  $\R_\ell=\{R_b \mid b\in B_\ell\}$. At each moment, we maintain the
  subset $\R'_\ell\subseteq \R_\ell$ of rectangles that do not contain
  $k$ points of $A_\ell$. The rectangles $R_b\in \R'_\ell$ are stored
  in a dynamic balanced binary search tree $T$ sorted by the value
  $b_y$. Moreover, for each rectangle $R_b\in\R'_\ell$ we also store
  a list $\L_b$ of the points of $A_\ell$ that it contains and the
  length of the list $\L_b$, that is, $|R_b\cap A_\ell|$.

  When the line $\ell$ arrives at a point $a\in A$, we find the $m$
  rectangles of $\R'_\ell$ that contain $a$. Traversing the tree $T$,
  this can be done in $O(m+\log n)$ time. For each of the $m$
  rectangles $R_b\in \R'_\ell$ that contain $a$, we add $a$ to the
  list $\L_b$. Moreover, if $\L_b$ now contains $k$ points, then
  $R_b$ does not belong to $\R'_\ell$ anymore and we remove the record
  from the tree $T$.
	
  When the line $\ell$ reaches a point $b\in B$, then $R_b$ becomes an
  element of $\R_\ell$ and we insert $R_b$ into $T$. If there is a
  point $a$ that belongs to $A$ and $B$, then we first consider it as
  a point of $B$ and then as a point of $A$. In this way $a$ becomes
  an element of $R_a$.
	
  Each insertion or deletion in $T$ takes $O(\log n)$. We make $|B|$
  insertions and at most $|B|$ deletions in $T$, for a total of
  $O(n\log n)$ time. For each point $a\in A$ we spend $O(\log n)$
  plus $O(1)$ time for each rectangle $R_b$ for which we report
  $a$. Thus, the running time is $O(kn+n\log n)$.
\end{proof}

For a set $Q$ of points, a point $q\in Q$, and a parameter~$k$ define
\begin{align*}
  \Phi(Q,q,k) := \min \ \{ \area(R) \ | \ &
  \text{$R$ is a box with $q\in \top(R)$ or $q\in \bot(R)$} \\
  & \text{and $R$ contains at least $k$ points of $Q$.} \ \}
\end{align*}

\begin{figure}
  \centerline{\includegraphics{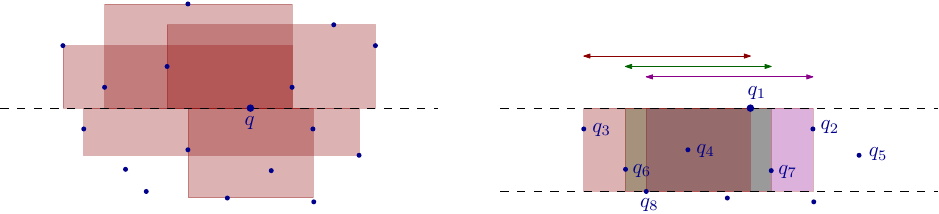}}
  \caption{Left: Example of the scenario considered in the definition
    of $\Phi(Q,q,k)$. 
    Some of the feasible boxes when $k=5$ are shown.
    Right: The boxes considered when $k=5$, $q\in
    \top(R)$ and we consider $Q_8 = \{q_1, \dots, q_8\}$. 
    Thus $q_8\in \bot(R)$. The span of the relevant boxes is
    shown with the arrows above.} 
  \label{fig:smallcase}	
\end{figure}

An example is shown in Figure~\ref{fig:smallcase}. We will reduce our
problem to many instances of the problem of computing $\Phi(Q,q,k)$
with $|Q|=O(k)$. We first discuss how to solve such instances.

\begin{lemma}
  \label{lem:smallcase}
  Given $Q$, $q$ and $k$, we can compute $\Phi(Q,q,k)$ in $O(|Q|^2)$ time.
\end{lemma}
\begin{proof}
  Let us discuss the case where $q\in \top(R)$, the other case being
  symmetric. Let $q_1, q_2, \dots, q_m$ be the points of~$Q$ whose
  $y$-coordinate is not larger than $q_y$, in decreasing order
  of~$y$-coordinate, and let $Q_i=\{q_1,\dots, q_i\}$.

  Once we have a sorted list with the elements of~$Q_i$ in increasing
  $x$-coordinate, then we can find in $O(|Q_i|) = O(i)$ time the
  minimum-area box $R$ that contains $k$ points with $q\in\top(R)$ and
  $q_i\in \bot(R)$, using a linear scan of the list with two pointers
  that are offset by $k$~elements. See Figure~\ref{fig:smallcase} for
  an example.

  We can therefore proceed as follows: We first compute the set~$Q_m$
  and sort it by $x$-coordinate, in time~$O(|Q| \log |Q|)$. We then
  repeatedly compute the best box for the current set~$Q_i$
  (initially~$i=m$) in time~$O(i)$, then delete from the list the
  element with the smallest~$y$-coordinate to obtain~$Q_{i-1}$, again
  in time~$O(i)$. The total running time is~$O(|Q|^2)$.
\end{proof}

For a set $P$ of points, a horizontal line $\ell$, and a parameter~$k$
define
\begin{align*}
  \Psi(P,\ell,k) := \min \ \{ \area(R) \ | \ &
  \text{$R$ is a box intersecting $\ell$} \\
  & \text{such that $R$ contains at least $k$ points of $P$} \ \}
\end{align*}

Recall that $\optarea(P,k)$ is the area of the optimal solution for
the original, global problem. Thus, it is obvious that
$\optarea(P,k)\le \Psi(P,\ell,k)$ for all $P$, $\ell$ and $k$. The
following lemma explains that when an optimal, global solution is
intersected by the line $\ell$, then we can reduce the search to a few
small problems of size $O(k)$.

\begin{lemma}
  \label{lem:fewsmallcases}
  Given $P$, $\ell$, and $k$, we can compute in $O(kn+ n \log n)$ time
  sets $Q_p\subseteq P$, indexed by $p\in P$, with the following
  properties:
  \begin{itemize}
  \item $Q_p$ has $O(k)$ points for each $p\in P$.
  \item For each $k'\le k$, if $\optarea(P,k')= \Psi(P,\ell,k')$, 
    then $\optarea(P,k')=\Phi(Q_p,p,k')$ for some $p\in P$.
  \end{itemize}
\end{lemma}
\begin{proof}
  For each $p\in P$ let $\overline{p}$ be the point symmetric to $p$
  with respect to the line $\ell$. For each point $q$ of the plane,
  $q\notin \ell$, we define the following objects. See
  Figure~\ref{fig:fewsmallcases1}.
  \begin{itemize}
  \item Let $\slab(q)$ be the horizontal slab 
    defined by $\ell$ and the line parallel to $\ell$ through $q$.
  \item Let $R^\shortrightarrow_q$ be the $3$-sided rectangle 
    $\slab(q)\cap \{ (x,y)\in \RR^2\mid x\ge q_x\}$ and let
    $P^\shortrightarrow_q$ be the $k$ points of $P$ with smallest
    $x$-coordinate inside $R^\shortrightarrow_q$. 
  \item Let $R^\shortleftarrow_q$ be the $3$-sided rectangle 
    $\slab(q)\cap \{ (x,y)\in \RR^2\mid x\le q_x\}$ and let
    $P^\shortleftarrow_q$ be the $k$ points of $P$ with largest
    $x$-coordinate inside $R^\shortleftarrow_q$. 
  \end{itemize}
  
  \begin{figure}
    \centerline{\includegraphics[page=1]{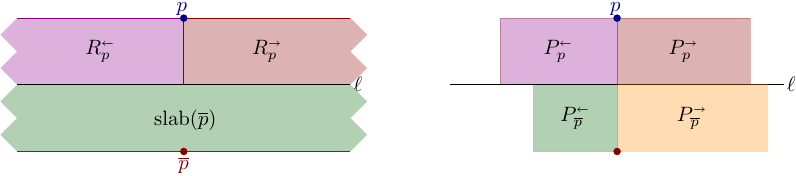}}
    \caption{Notation used in Lemma~\ref{lem:fewsmallcases}. On the
      right side we show the portions of the $3$-sided rectangles that
      contain $Q_p$.}
    \label{fig:fewsmallcases1}	
  \end{figure}

  For each $p\in P$, we define $Q_p$ as the union of
  $P^\shortrightarrow_p$, $P^\shortleftarrow_p$,
  $P^\shortrightarrow_{\overline p}$, and
  $P^\shortleftarrow_{\overline p}$. It is clear that each set $Q_p$
  has at most $4k$ points, so the first property of the lemma holds.
	
  To show the second property, consider any fixed $k'\le k$ and assume
  that $\optarea(P,k')= \Psi(P,\ell,k')$. Then there exists an
  optimal box~$R^*$ with $\area(R^*)=\optarea(P,k')$ such that~$R^*$
  intersects~$\ell$. Let $t^*$ and $b^*$ be points of $P$ on
  $\top(R^*)$ and $\bot(R^*)$, respectively. Assume without loss of
  generality that the distance from $t^*$ to $\ell$ is at least the
  distance from $b^*$ to $\ell$. This means that $R^*\cap P$ is
  contained in $\slab(t^*)\cup \slab(\overline{t^*})$.

  We next show that $R^*\cap P$ is contained in $Q_{t^*}$. Assume,
  for the sake of reaching a contradiction, that $R^*\cap P$ contains
  a point $a$ that is not in $Q_{t^*}$. See
  Figure~\ref{fig:fewsmallcases2}. Therefore, $a$ is contained in one
  of the $3$-sided rectangles used to define $Q_{t^*}$, namely
  $R^\shortrightarrow_{t^*}$, $R^\shortleftarrow_{t^*}$,
  $R^\shortrightarrow_{\overline{t^*}}$,
  $R^\shortleftarrow_{\overline{t^*}}$. Let $\tilde R\in
  \{R^\shortrightarrow_{t^*}, R^\shortleftarrow_{t^*},
  R^\shortrightarrow_{\overline{t^*}},
  R^\shortleftarrow_{\overline{t^*}}\}$ be the $3$-sided rectangle
  that contains $a$, let $\tilde P\in \{P^\shortrightarrow_{t^*},
  P^\shortleftarrow_{t^*}, P^\shortrightarrow_{\overline{t^*}},
  P^\shortleftarrow_{\overline{t^*}}\}$ be the set contained in
  $\tilde R$ and let $\tilde q$ be the point of $\tilde P$ furthest
  from the vertical line through $t^*$. Note that $\tilde P$ contains
  $k$ points, as otherwise there cannot be any point of $P$ in $\tilde
  R\setminus \tilde P$ and $a$ cannot exist. By the way we selected
  the points of $\tilde P$ inside $\tilde R$ we have
  \[
  |t^*_x-\tilde q_x| ~<~ |t^*_x-a_x|. 
  \]
  Here we are using general position to rule out the possibility of
  equality. Note that the bounding box $\bb(\tilde P)$ of $\tilde P$
  contains $k\ge k'$ points and has area at most
  \[
  |t^*_x-\tilde q_x|\cdot \dist(t^*,\ell),
  \]
  where $\dist(t^*,\ell)$ denotes the vertical distance from $t^*$ to
  the line $\ell$. On the other hand, since $R^*$ intersects $\ell$
  and has $a$ and $t^*$ in its boundary, we have
  \[
  \area(R^*) ~\ge~ |t^*_x-a_x|\cdot \dist(t^*,\ell) 
  ~>~ |t^*_x-\tilde q_x|\cdot \dist(t^*,\ell) ~\ge~ \area(\bb(\tilde P)).
  \]
  This contradicts the optimality of $R^*$ for covering $k'$ points.
  This finishes the proof that $R^*\cap P$ is contained in $Q_{t^*}$,
  and therefore the second property holds.
  
  \begin{figure}
    \centerline{\includegraphics[page=2]{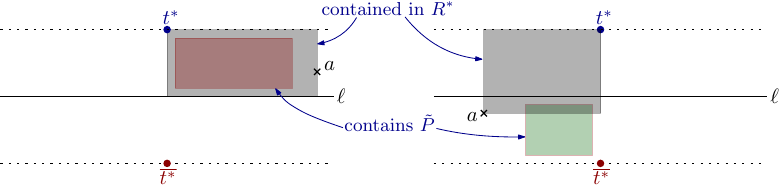}}
    \caption{Part of the proof of Lemma~\ref{lem:fewsmallcases} where we show
      that a point $a$ outside $R^*\cap Q_{t^*}$ cannot exist.
      On the left we have the case when $\tilde P=P^\shortrightarrow_{t^*}$
      and on the right the case $\tilde P=P^\shortleftarrow_{\overline{t^*}}$.}
    \label{fig:fewsmallcases2}	
  \end{figure}
  
  It remains to show that the construction of the sets $Q_p$, for all
  $p\in P$, can be done in $O(kn+n\log n)$ time. For this we use
  Lemma~\ref{lem:rangek} a few times, as follows. Let $P^+$ and $P^-$
  be the points above and below $\ell$, respectively. We also define
  $\overline{P^+}=\{ \overline{p} \mid p\in P^+\}$ and
  $\overline{P^-}=\{ \overline{p} \mid p\in P^-\}$. The point sets
  $P^\shortrightarrow_q$ for all $q\in P^- \cup \overline{P^+}$, are
  obtained using Lemma~\ref{lem:rangek} with $A=P^-$ and $B=P^- \cup
  \overline{P^+}$. The sets $P^\shortrightarrow_q$ for all $q\in P^+
  \cup \overline{P^-}$, the sets $P^\shortleftarrow_q$ for all $q\in
  P^- \cup \overline{P^+}$, and the sets $P^\shortleftarrow_q$ for all
  $q\in P^+ \cup \overline{P^-}$, are computed in a similar way, using
  symmetric versions of Lemma~\ref{lem:rangek}.
\end{proof}

\begin{lemma}
  \label{lem:mergestep}
  Let $P$ be a set of $n$ points, let $\ell$ be a horizontal line, and
  let $k$ be a positive integer. After $O(nk+ n \log n)$
  preprocessing time, we can compute, for any given $k'\le k$, a value
  $V(P,\ell,k')$ with the following properties in $O(nk^2)$ time:
  \begin{itemize}
  \item $\optarea(P,k') \le V(P,\ell,k')$; 
  \item if $\optarea(P,k')= \Psi(P,\ell,k')$, then
    $V(P,\ell,k')=\optarea(P,k')$. 
  \end{itemize}
\end{lemma}
\begin{proof}
  We compute the sets $Q_p$, indexed by $p\in P$, using
  Lemma~\ref{lem:fewsmallcases}. This finishes the preprocessing and
  takes time $O(kn+ n \log n)$ time.
  
  Now suppose that we are given a value $k'\le k$. For each $p\in P$,
  we use Lemma~\ref{lem:smallcase} to find the value $\Phi(Q_p,p,k')$
  in $O(|Q_p|^2)=O(k^2)$ time. We return the value $V(P,\ell,k')=\min
  \{ \Phi(Q_p,p,k') \mid p\in P \}$. The computation takes $O(nk^2)$
  time.
	
  Since for each $p\in P$ the value $\Phi(Q_p,p,k')$ is the area of a
  box containing $k'$ points of~$P$, we have $V(P,\ell,k')\ge
  \optarea(P,k')$. If $\optarea(P,k')= \Psi(P,\ell,k')$, then
  Lemma~\ref{lem:fewsmallcases} guarantees that $\optarea(P,k')=
  \Phi(Q_{p_0},p_0,k')$ for some $p_0\in P$, and therefore
  \[ 
  \optarea(P,k') \; = \; \Phi(Q_{p_0},p_0,k')
  \; \ge \; \min \{ \Phi(Q_p,p,k') \mid p\in P \} 
  \; = \; V(P,\ell,k').
  \]
  We conclude that $V(P,\ell,k')=\optarea(P,k')$.
\end{proof}

\begin{theorem}
  \label{thm:main1}
  Given a set of $n$ points $P$ and a value $k$, we can preprocess $P$
  in $O(nk\log n+ n \log^2 n)$ time such that, for any given $k'\le
  k$, we can find in $O(nk^2\log n)$ time a minimum-area box that
  contains at least $k'$~points of~$P$.
\end{theorem}
\begin{proof}
  Consider a horizontal line $\ell$ such that at most half of the
  points of $P$ are above $\ell$ and at most half of the points are
  below~$\ell$. Let $P^+$ and $P^-$ be the subset of $P$ above and
  below~$\ell$, respectively. For any number of points $k'$, where
  $1\le k' \le n$ we have
  \[
  \optarea(P,k') ~=~ \min \; \big\{ \ \Psi(P,\ell,k'),
  \ \optarea(P^+,k'), \ \optarea(P^-,k') \ \big\}. 
  \]
  Indeed, an optimal solution containing $k'$ points is either
  intersected by $\ell$ or it contains points from only one of the
  sets $P^+$ and $P^-$. This is the basis for an algorithm based on
  divide and conquer.
  
  In the preprocessing, we use Lemma~\ref{lem:mergestep} for $P$,
  $\ell$ and $k$, which takes $O(nk+n\log n)$ time, and then
  recursively preprocess $P^+$ and $P^-$. Since the recursion has
  $\log n$ levels and since any two point sets at the same level of
  the recursion are disjoint, we spend $O(nk\log n+ n \log^2 n)$ time
  in preprocessing.

  When we are given a value $k'$, we compute $\optarea(P,k')$ using
  the same recursive pattern. At the first level, with the point set
  $P$ and the line $\ell$, we spend $O(nk^2)$ time to compute
  $V(P,\ell,k')$, using Lemma~\ref{lem:mergestep}. Then we go on to
  compute $\optarea(P^+,k')$ and $\optarea(P^-,k')$ recursively, using
  our already-done preprocessing. Finally, we return the minimum of
  $V(P,\ell,k')$, $\optarea(P^+,k')$ and $\optarea(P^-,k')$. By
  Lemma~\ref{lem:mergestep}, we always have $\optarea(P,k')\le
  V(P,\ell,k')$ and, when $\optarea(P,k')= \Psi(P,\ell,k')$, we also
  have $\optarea(P,k')=V(P,\ell,k')$. It follows that
  \[
  \optarea(P,k') ~=~ \min \; \big\{ \; V(P,\ell,k'), \;
  \optarea(P^+,k'),\; \optarea(P^-,k') \; \big\} 
  \]
  and thus we are returning the correct value of $\optarea(P,k')$.
  Since we have $\log n$ levels in the recursion, we spend $O(nk^2 \log n)$ time in total.
\end{proof}

\begin{corollary}
  \label{coro:main1}
  Given a set of $n$ points $P$ and a value $k$ we can find in $O(n
  k^2\log n+ n \log^2 n)$ time a minimum-area box that contains at
  least $k$~points of~$P$.
\end{corollary}
\begin{proof}
  We apply Theorem~\ref{thm:main1} with $k'=k$. In this scenario we
  can get rid of the preprocessing step and, at each level of the
  recursion, compute the values $\Phi(Q_p,p,k)$ immediately after
  generating the sets $Q_p$.
\end{proof}

\section{Maximizing the number of points for a given area}
\label{sec:points}

We now turn to the problem of finding the maximum number of points
that can be covered by a box of a area~$\alpha > 0$. As mentioned in
the introduction, let $\optpoints(P,\alpha)$ be this number of points.
We first compute a constant-factor approximation to
$\optpoints(P,\alpha)$. Then we explain how to obtain a
$(1+\eps)$-approximation using an algorithm whose running time depends
on the value $\optpoints(P,\alpha)$. Finally, we use random sampling
to get a $(1+\eps)$-approximation to $\optpoints(P,\alpha)$ in
near-linear time for a fixed~$\eps>0$.

\subsection{A 4-approximation algorithm}
\label{ssec:4apx}

For a horizontal line~$\ell$ and a point~$p\notin \ell$, let
$R_\alpha^\shortrightarrow(p,\ell)$ be the box that has area $\alpha$,
has $p$~as a corner, has an edge contained in $\ell$, and contains
points with $x$-coordinates larger than~$p_x$. Let
$R_\alpha^\shortleftarrow(p,\ell)$ be the box defined in the same way,
but with points with $x$-coordinates smaller than~$p_x$. See
Figure~\ref{fig:4approx}, left. Let $\R_\alpha(\ell)$ be the set of
boxes $\bigcup_{p\in P}\{ R_\alpha^\shortrightarrow(p,\ell),
R_\alpha^\shortleftarrow(p,\ell)\}$. Let $\optpoints(P,\ell,\alpha)$
be the maximum number of points of $P$ covered by a box of
area~$\alpha$ that intersects the line~$\ell$.

\begin{figure}
  \centerline{\includegraphics{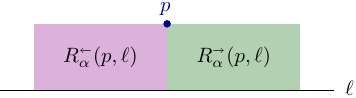}}
  \caption{The boxes $R_\alpha^\shortleftarrow(p,\ell)$ and
    $R_\alpha^\shortrightarrow(p,\ell)$, the boxes have area $\alpha$.}
  \label{fig:4approx}	
\end{figure}
    
\begin{lemma}
  \label{lem:feasible}
  There is some $R\in \R_\alpha(\ell)$ such that $|P\cap R|\ge
  \optpoints(P,\ell,\alpha)/4$. 
\end{lemma}
\begin{proof}
  Let $R^*$ be a box of area $\alpha$ covering
  $\optpoints(P,\ell,\alpha)$ points and intersecting~$\ell$. Let
  $\ell^+$ and $\ell^-$ denote the closed half-planes above and below
  $\ell$, respectively. Define $R^+ := \R^*\cap\ell^+$ and $R^- :=
  \R^*\cap\ell^-$, and assume without loss of generality that
  $|R^+\cap P| \geq |R^* \cap P|/2$. Let $p$ be the point in $P\cap
  R^*$ with maximum $y$-coordinate. Then $R^+ \subset
  R_\alpha^\shortrightarrow(p,\ell)\cup
  R_\alpha^\shortleftarrow(p,\ell)$. Hence, at least one of
  $R_\alpha^\shortrightarrow(p,\ell)$ or
  $R_\alpha^\shortleftarrow(p,\ell)$ must contain $|R^+ \cap P|/2 \geq
  |R^* \cap P|/4$ points.
\end{proof}

\begin{theorem}
  \label{thm:4approx}
  Given a set of $n$ points $P$ and a value $\alpha> 0$, we can
  compute in $O(n\log^2 n)$ time a value $\kappa(P,\alpha)$ such that
  $\optpoints(P,\alpha)/4 \le \kappa(P,\alpha) \le
  \optpoints(P,\alpha)$.
\end{theorem}
\begin{proof}
  We preprocess $P$ for $4$-sided rectangle counting
  queries~\cite{Willard85}. The preprocessing takes $O(n\log n)$ time
  and for each query rectangle $R$ we can report $|R\cap P|$ in
  $O(\log n)$ time.
	
  Then we proceed with a recursive algorithm. Take a line $\ell$ that
  splits $P$ into two sets $P^+$ and $P^-$ of roughly equal size.
  Note that
  \[
  \optpoints(P,\alpha) ~=~ \max \; \big\{\; \optpoints(P^+,\alpha), \;
  \optpoints(P^-,\alpha), \; \optpoints(P,\ell,\alpha) \; \big\}. 
  \]	
  
  We build the set of boxes $\R_\alpha(\ell)$ in $O(n)$ time. For
  each box $R\in \R_\alpha(\ell)$ we query the data structure to
  obtain~$|R\cap P|$. Thus, we obtain $\kappa(P,\ell,\alpha)=\max \{
  |R\cap P| \mid R\in \R_{\alpha}(\ell) \}$ in $O(n\log n)$ time. By
  Lemma~\ref{lem:feasible}, we have $\optpoints(P,\ell,\alpha)/4 \le
  \kappa(P,\ell,\alpha)\le \optpoints(P,\ell,\alpha)$. Then, we
  return the best between the value $\kappa(P,\ell,\alpha)$ and the
  values $\kappa(P^+,\alpha),\kappa(P^-,\alpha)$ obtained recursively
  for $P^+$ and $P^-$, respectively. Since at each level of the
  recursion the point sets being considered are disjoint, we spend
  $O(n\log n)$ time at each level of the recursion, for a total
  $O(n\log^2 n)$ for the whole algorithm.

  Since the algorithm only considers boxes of area $\alpha$, the
  returned value is obviously at most $\optpoints(P,\alpha)$. On the
  other hand, by induction we have $\kappa(P^+,\alpha)\ge
  \optpoints(P^+,\alpha)/4$ and $\kappa(P^-,\alpha)\ge
  \optpoints(P^-,\alpha)/4$. Together with $\kappa(P,\ell,\alpha)\ge
  \optpoints(P,\ell,\alpha)/4$ we obtain that
  \begin{align*}
    \kappa(P,\alpha)~&=~
    \max \{\kappa(P^+,\alpha), \,\kappa(P^-,\alpha),
    \,\kappa(P,\ell,\alpha)\} \\ 
    &\ge~
    \max \{\optpoints(P^+,\alpha)/4, \,\optpoints(P^-,\alpha)/4,
    \,\optpoints(P,\ell,\alpha)/4\} \\  
    ~&=~ \optpoints(P,\alpha)/4. \qedhere
  \end{align*}
\end{proof}

\subsection{Properties of random sampling}
\label{ssec:rsprops}

For the rest of this section, let $P$ be a set of $n$ points in the
plane, and let $\eps$ be a real value with $0< \eps< 1$. We use
relative approximations~\cite{Har-PeledS11} to show that a random
sample from $P$ can be used to count the points of $P$ inside each
box, assuming that the box has enough points. We use $s$ for the
cardinality of the sample.

\begin{lemma}
  \label{lem:sampling}
  Suppose that $\kappa$ satisfies $\optpoints(P,\alpha)\le \kappa$.
  Let $s=\min\{ n, \frac{c}{\eps^2} \frac{n}{\kappa}\log n\}$, where
  $c$ is an appropriate absolute constant, and let $S$ be a random
  sample of $P$ with $s$ points. Then with probability at least
  $1-1/n$ the following properties hold simultaneously:
  \begin{itemize}
  \item For each box~$R$ of area at most~$\alpha$
    \[
    \left| \frac{|P\cap R|}{n} - \frac{|S\cap R|}{s}\right| ~\le~
    \eps\cdot \frac{\kappa}{n}; 
    \]
  \item $\optpoints(S,\alpha)=O\big((1/\eps^2)\log n\big)$.
  \end{itemize}
\end{lemma}
\begin{proof}
  One can prove this using Chernoff bounds as we did in our first
  preprint~\cite{BergCCEK16}. A more compact proof uses relative
  approximations, as described next.
	
  We consider the case when $s = \frac{c}{\eps^2}\frac{n}{\kappa}\log
  n$. For the other case we have $S=P$ and $\kappa<(c/\eps^2)\log n$,
  so the claims trivially hold.

  Let $\R$ be the family of all boxes of area at most~$\alpha$. A
  subset $S\subseteq P$ is a \emph{relative
    $(\rho,\eps)$-approximation} for $(P,\R)$ if
  \[
  \forall R\in \R:~~~ 
  \left| \frac{|P\cap R|}{|P|}- \frac{|S\cap R|}{|S|} \right| ~\le~ 
  \eps\cdot \max\left\{\frac{|P\cap R|}{|P|},\rho \right\}.
  \]
  \mbox{Har-Peled} and Sharir~\cite[Theorem 2.11]{Har-PeledS11} show that the
  results of Li et al.~\cite{LiLS01} imply the following: A random
  sample of $P$ of size
  \[
  \frac{c'}{\eps^2\rho}\left( \delta \log \frac{1}{\rho}+\log
  \frac{1}{q}\right), 
  \]
  where $c'$ is an appropriate absolute constant, is a relative
  $(\rho,\eps)$-approximation for $(P,\R)$ with probability at least
  $1-q$. Here $\delta$ is the VC-dimension of the range space
  $(P,\R)$; in our case~$\delta \leq 4$. Setting $\rho=\kappa/n$ and
  $q=1/n$, we obtain that a random sample of size
  \[
  s= \frac{c'}{\eps^2\rho}\left( \delta \log \frac{1}{\rho}+\log
  \frac{1}{q}\right) ~\le~ 
  \frac{c'n}{\eps^2\kappa}\left( 4 \log \frac{n}{\kappa}+\log n\right)~\le~
  \frac{5 c'}{\eps^2\kappa} n \log n
  \]
  is a relative $(\kappa/n,\eps)$-approximation for $(P,\R)$ with
  probability at least~$1-1/n$. The constant~$c$ in the statement of
  the lemma is then~$5c'$.
  
  It remains to show that, if $S$ is a relative
  $(\kappa/n,\eps)$-approximation for $(P,\R)$, then both properties
  in the lemma hold. Since $S$ is a relative
  $(\kappa/n,\eps)$-approximation we have
  \[
  \forall R\in \R:~~~ 
  \left| \frac{|P\cap R|}{n}- \frac{|S\cap R|}{s} \right| ~\le~ 
  \eps\cdot \max\left\{\frac{|P\cap R|}{n},\frac{\kappa}{n} \right\} = 
  \eps\cdot \frac{\kappa}{n},
  \]
  where in the last step we used $|P\cap R|\le \optpoints(P,\alpha)\le
  \kappa$. This shows the first item.
  
  For the second item we note that, for any box~$R$ of area at
  most~$\alpha$, we have
  \begin{align*}
    |S\cap R|~\le~ s\big(\frac{|P\cap R|}{n} + \eps \cdot \frac{\kappa}{n}\big)
    ~\le~ 
    \frac{s}{n} \big( \kappa+ \eps \kappa \big)
    = (1+\eps) \frac{s\kappa}{n}
    = (1+\eps) \frac{c}{\eps^2}\log n.
  \end{align*}
  It follows that $\optpoints(S,\alpha)=O\big((1/\eps^2)\log n\big)$.
\end{proof}

\subsection[Approximation algorithm]{A $(1-\eps)$-approximation algorithm}
\label{ssec:appxalg}

We start by giving an output-sensitive $(1-\eps)$-approximation
algorithm whose running time depends quadratically on the size of the
output.
\begin{lemma}
  \label{lem:smallk}
  Given a set $P$ of $n$ points, a value $\alpha> 0$, and a parameter
  $\eps$ with $0<\eps<1$, we can compute in $O\left( n(\kappa^*)^2\log
  n \log (1/\eps) + n\log^2 n\right)$ time a box~$R$ of area at most
  $\alpha$ that covers at least $(1-\eps)\kappa^*$ points of $P$,
  where $\kappa^*=\optpoints(P,\alpha)$.
\end{lemma}
\begin{proof}
  Using Theorem~\ref{thm:4approx} we compute a $4$-approximation value
  $\kappa_a$ satisfying $\kappa^*/4 \le \kappa_a \le \kappa^*$. We
  apply Theorem~\ref{thm:main1} with the value $4\kappa_a$, which is
  an upper bound for $\kappa^*$. We spend $O(n\kappa_a \log n+n\log^2
  n)= O(n\kappa^* \log n+n\log^2 n)$ time in the preprocessing and
  then, for any given $k'\le \kappa_a$, we can compute
  $\optarea(P,k')$ in $O\left( n(\kappa^*)^2\log n \right)$ time.
	
  Consider the set $K$ of values of the form $\kappa_a+i\cdot \eps
  \kappa_a$, with $i\in \NN$, inside the interval
  $[\kappa_a,4\kappa_a]$. We perform binary search on~$K$ to find the
  value $\tilde k\in K$ such that $\optarea(P,\tilde k)\le \alpha$ but
  $\optarea(P,\tilde k+\eps \kappa_a)> \alpha$. We then have
  \[	
  \kappa_a ~\le~ \tilde k ~\le \optpoints ~\le~ \tilde k +\eps \kappa_a
  ~\le~ \tilde k +\eps \optpoints,
  \]
  and thus 
  \[
  \tilde k ~\ge~ \optpoints -  \eps \optpoints ~=~
  (1-\eps)\optpoints.
  \]
  
  Since $K$ has $O(1/\eps)$ values, the binary search performs $O(\log
  (1/\eps)$ steps. At each step we have to compute $\optarea(P,k')$
  for some value $k'\le 4\kappa_a$, which takes $O\left(
  n(\kappa^*)^2)\log n \right)$ time. In total, we spend $O\left(
  n(\kappa^*)^2)\log n \log (1/\eps)\right)$ time after $O(n\kappa^*
  \log n+n\log^2 n)$ preprocessing time.
\end{proof}

\begin{theorem}
  \label{thm:main3}
  Given a set of $n$ points $P$ in the plane and a value $\alpha> 0$,
  let $\optpoints(P,\alpha)$ be the maximum number of points from $P$
  that can be covered with a box of area~$\alpha$. Given a
  parameter~$\eps$, where $0\le \eps\le 1/2$, with probability at
  least $1-1/n$ we can find in $O\left((n/\eps^4) \log^3 n
  \log(1/\eps)\right)$ time a box~$\tilde R$ of area $\alpha$ that
  covers at least $(1-\eps)\optpoints(P,\alpha)$ points from $P$.
\end{theorem}
\begin{proof}
  Using Theorem~\ref{thm:4approx} we compute in $O(n\log^2 n)$ time
  a value $\kappa_a$ satisfying 
  \[
  \optpoints(P,\alpha)/4 ~\le~ \kappa_a ~\le~ \optpoints(P,\alpha).
  \]
  Set $\kappa=4\kappa_a$, so that $\optpoints(P,\alpha)\le \kappa$,
  set $s=\min\{ n, \frac{c}{\eps^2}\frac{n}{\kappa}\log n\}$, and take
  a sample~$S$ of $P$ with~$s$ points. Henceforth we assume that $S$
  satisfies the properties of Lemma~\ref{lem:sampling}, which occurs
  with probability at least $1-1/n$.
  
  Using Lemma~\ref{lem:smallk} for the sample $S$, we compute a
  box~$\tilde R$ of area $\alpha$ covering at least $(1-\eps)
  \optpoints(S,\alpha)$ points of~$S$. We return~$\tilde R$.
	
  Let us analyze the running time of the algorithm. By
  Lemma~\ref{lem:sampling}, we have
  $\optpoints(S,\alpha)=O\big((1/\eps^2)\log n\big)$. This means that the
  algorithm of Lemma~\ref{lem:smallk} takes time
  \[
  O\left(|S| \left(\optpoints(S,\alpha)\right)^2\log |S| \log(1/\eps)
  + |S|\log^2 |S| \right).
  \]
  Substituting the value $\optpoints(S,\alpha)=O\big((1/\eps^2)\log n\big)$
  and $|S|\le n$ we get the time bound
  \[
  O\left(|S| \left((1/\eps^2)\log n\right)^2 \log |S| \log(1/\eps) +
  |S|\log^2 |S| \right) ~=~ 
  O\left((n/\eps^4) \log^3 n  \log(1/\eps)\right).
  \]
  
  Finally, we analyze the approximation error. Let $R^*$ be an
  optimal solution for $P$: A box of area $\alpha$ that contains
  $\optpoints(P,\alpha)$ points of~$P$. Since $\tilde R$ covers
  $(1-\eps) \optpoints(S,\alpha)$ points of the sample~$S$,
  Lemma~\ref{lem:sampling} implies that
  \begin{align*}
    \frac{|P\cap \tilde R|}{n} ~&\ge~
    \frac{|S\cap \tilde R|}{s} - \eps \frac{\kappa}{n}\\
    &\ge~ \frac{(1-\eps) \optpoints(S,\alpha)}{s} - \eps \frac{\kappa}{n}\\
    &\ge~ (1-\eps) \frac{|S\cap R^*|}{s} - \eps \frac{\kappa}{n}\\
    &\ge~ (1-\eps) \frac{|P\cap R^*|}{n} - \eps \frac{\kappa}{n} -
    \eps \frac{\kappa}{n}\\ 
    &=~ (1-\eps) \frac{\optpoints(P,\alpha)}{n} - 2 \eps \frac{\kappa}{n}.
  \end{align*}
  This means that 
  \[
  |P\cap \tilde R| \ge (1-\eps) \optpoints(P,\alpha)- 2 \eps \kappa,
  \]
  and using that $\kappa=4\kappa_a \le 4\optpoints(P,\alpha)$, we get
  \[
  |P\cap \tilde R| \ge (1-\eps) \optpoints(P,\alpha)- 2 \eps
  (4\optpoints(P,\alpha)) \ge (1-9\eps) \optpoints(P,\alpha). 
  \]
  Repeating the analysis with $\eps'=\eps/9$ in place of $\eps$, the
  result follows.
\end{proof}

\subsection{Can we make the algorithm deterministic?}
\label{sec:deterministic}

To make our approximation algorithm deterministic, we would need a
deterministic construction of a relative $(\rho,\eps)$-approximation
with respect to boxes. It is currently unclear if a relative
approximation of the desired size can be computed deterministically in
an efficient manner. Another option would be to use
$\eps$-approximations for boxes. Given a set of points $P$ in the
plane, a subset $A\subseteq P$ is a $\delta$-approximation\footnote{We
  use $\delta$ rather than~$\eps$ here to avoid confusion with the
  different roles of~$\eps$.} with respect to boxes if
\[
\forall \text{~boxes $R$}:~~~ 
\left| \frac{|R\cap P|}{|P|} - \frac{|R\cap A|}{|A|} \right| ~\le~ \delta.
\]
$\delta$-approximations are good for counting the number of points
inside each box with an error of~$\delta|P|$. After we have a
constant-factor approximation~$\kappa_a$ to the
value~$\optpoints(P,\alpha)$, we could thus use a
$\delta$-approximation with $\delta=\eps\kappa_a$. There are
$\delta$-approximations with respect to boxes of size
roughly~$O(1/\delta)$, which would be better than the random sample we
are currently using. However, constructing such a
$\delta$-approximation takes roughly $O(n/\delta^3)$ time; see
Phillips~\cite{Phillips08} for the latest results. For example, when
$\kappa_a=\Theta(\optpoints)=\Theta(\sqrt{n})$, this means that we
need roughly $\tilde O(n^{5/2})$ time. Thus, building
$\delta$-approximations deterministically in near-linear time is the
current bottleneck for this approach.

\section{Conclusions}

There are several avenues that can be pursued to improve our results:

Improving Lemma~\ref{lem:smallcase} directly improves our time bounds
for computing $\optarea(P,k)$. One approach would be to reduce the
problem of Lemma~\ref{lem:smallcase} to the following problem:
Maintain a set of $O(k)$ points on the real line under insertions such
that, after each insertion, we can recover the smallest interval that
contains $k$ of the points. Moreover, we know the order of the
insertions in advance. If we can handle each insertion in $o(k)$
time, then the result in Lemma~\ref{lem:smallcase} can be improved,
and consequently $\optarea(P,k)$ can be computed faster.

One may also try to compute the $k$ values $\optarea(P,1)$, $\dots$,
$\optarea(P,k)$ faster than using the algorithm for each $k'\in
\{1,\dots,k\}$ independently. In particular, if in
Lemma~\ref{lem:smallcase} we could compute all the values
$\Phi(Q,q,1), \Phi(Q,q,2), \dots, \Phi(Q,q,|Q|)$ in $o(|Q|^3)$ time,
then a better algorithm could be obtained for this problem.

The following additional open problems remain:
\begin{itemize}
\item Is it possible to derandomize the algorithm described in
  Section~\ref{sec:points} (see the discussion in
  Section~\ref{sec:deterministic})?.
\item In $\RR^3$, can we find the smallest box covering $k$ points in
  time roughly $O(nk^3)$?  Note that any running time of the form
  $\tilde O(nk^c)$, for some constant $c$, would lead to a
  near-linear-time randomized $(1-\eps)$-approximation algorithm for
  the dual problem of covering as many points as possible with a box
  of given volume.
\item Is there a non-trivial lower bound, such as $\Omega(nk)$, for
  computing $\optarea(P,k)$ exactly?
\end{itemize}

\section*{Acknowledgments}

Some parts of this work have been done at the Fourth Annual Workshop
on Geometry and Graphs, held at the Bellairs Research Institute in
Barbados, March~6--11, 2016. The authors are grateful to the
organizers and to the participants of the workshop, especially to Luis
Barba for suggesting the problem. We would also like to thank Xavier
Goaoc for fruitful discussions on the subject.

{\raggedright
\bibliographystyle{plain}
\bibliography{biblio}}
\end{document}